\documentclass[10pt,conference]{IEEEtran}
\usepackage[latin1]{inputenc}
\usepackage{amsmath,hyperref}
\usepackage{pstricks,pst-node,multido}
\usepackage{eucal,ifthen}
\usepackage[expansion=false]{microtype}

\newlength{\tmpxa}
\newlength{\tmpxb}
\newlength{\tmpxc}

\newlength{\tmpya}
\newlength{\tmpyb}
\newlength{\tmpyc}

\SpecialCoor

\psset{unit=1.3cm}

\psset{%
    gridcolor=gray,
    subgridcolor=lightgray,
    gridwidth=0.2pt,
    subgridwidth=0.1pt}

\newcommand{\mypssources}{\psset{%
    radius=1.5mm,
    labelsep=3mm}}
\newcommand{\mypstargets}{\psset{%
    framesize=3mm,
    labelsep=3mm}}
\newcommand{\mypslinks}{\psset{%
    npos=0.65,
    labelsep=1pt,
    border=1pt}}

\newcommand{\mypsiedges}{\psset{%
    arcangle=0.01,
    linewidth=1.5pt,
    labelsep=2pt}}
\newcommand{\mypskedges}{\psset{%
    arcangle=0.01,
    linewidth=1.5pt,
    linestyle=dashed,
    dash=3pt 3pt,
    labelsep=2pt}}
\newcommand{\mypshyperedges}{\psset{%
    linearc=0.05,
    cornersize=absolute}}
\newcommand{\mypsihyperedges}{\mypshyperedges\psset{%
    linewidth=0.5pt}}
\newcommand{\mypskhyperedges}{\mypshyperedges\psset{%
    linewidth=0.5pt,
    linestyle=dashed,
    dash=2pt 2pt}}
\newcommand{\mypshypernodes}{\psset{%
    radius=2.5pt,
    labelsep=4pt}}
\newcommand{\mypsnodes}{\psset{%
    radius=2.5pt,
    labelsep=6pt}}

\newcommand{\myframe}[8]{%
    \pssetlength{\tmpxa}{#1}
    \pssetlength{\tmpya}{#2}
    \pssetlength{\tmpxb}{#3}
    \pssetlength{\tmpyb}{#4}
    \psaddtolength{\tmpxa}{-#5}
    \psaddtolength{\tmpya}{-#6}
    \psaddtolength{\tmpxb}{#7}
    \psaddtolength{\tmpyb}{#8}
    \psframe(\tmpxa,\tmpya)(\tmpxb,\tmpyb)
}%
\newcommand{\myframeL}[8]{%
    \pssetlength{\tmpxa}{#1}
    \pssetlength{\tmpya}{#2}
    \pssetlength{\tmpxb}{#1}
    \pssetlength{\tmpyb}{#2}
    \pssetlength{\tmpxc}{#3}
    \pssetlength{\tmpyc}{#4}
    \psaddtolength{\tmpxa}{-#5}
    \psaddtolength{\tmpya}{-#6}
    \psaddtolength{\tmpxb}{#7}
    \psaddtolength{\tmpyb}{#8}
    \psaddtolength{\tmpxc}{#7}
    \psaddtolength{\tmpyc}{#8}
    \pspolygon(\tmpxa,\tmpya)(\tmpxa,\tmpyc)(\tmpxb,\tmpyc)(\tmpxb,\tmpyb)(\tmpxc,\tmpyb)(\tmpxc,\tmpya)
}%
\newcommand{\myframeb}[6]{%
    \myframe{#1}{#2}{#3}{#4}{#5}{#5}{#6}{#6}%
}
\newcommand{\myframebL}[6]{%
    \myframeL{#1}{#2}{#3}{#4}{#5}{#5}{#6}{#6}%
}

\newcommand{\mysize}[1]{{\lvert #1 \rvert}}

\newcommand{\mybound}[2]{\Delta^{#1}_{#2}}

\newcommand{\myutil}{\omega}
\newcommand{\myoptutil}{\omega^{*}}
\newcommand{\myoptx}{x^{*}}
\newcommand{\myE}{\mathcal{E}}
\newcommand{\myG}{\mathcal{G}}
\newcommand{\myH}{\mathcal{H}}

\newcommand{\myKdash}{$K\mspace{-2.5mu}$-}
\newcommand{\myIdash}{$I\mspace{-1.5mu}$-}
\newcommand{\myXdash}{$X\mspace{-2.5mu}$-}
\newcommand{\myYdash}{$Y\mspace{-4mu}$-}
\newcommand{\Kadjacent}{\mbox{\myKdash adjacent}}
\newcommand{\Iadjacent}{\mbox{\myIdash adjacent}}
\newcommand{\Klength}{\mbox{\myKdash length}}
\newcommand{\Iedge}{\mbox{\myIdash edge}}
\newcommand{\Kedge}{\mbox{\myKdash edge}}
\newcommand{\Iedges}{\mbox{\myIdash edges}}
\newcommand{\Kedges}{\mbox{\myKdash edges}}
\newcommand{\Xedge}{\mbox{\myXdash edge}}
\newcommand{\Yedge}{\mbox{\myYdash edge}}
\newcommand{\Khop}{\mbox{\myKdash hop}}

\newcommand{\myparty}[2]{\overset{{\scriptstyle \text{($k_{#1}$)}}}{\rule{0pt}{1.8ex}#2}}
\newcommand{\myresource}[1]{\qquad{\scriptstyle \text{\raisebox{1pt}{($i_{#1}$)}}}}
\newcommand{\mylist}[3]{\multido{\i=1+1}{#1}{%
    \ifthenelse{\equal{\i}{1}}{}{#2}%
    #3_{\i}%
}}
\newcommand{\mypartylist}[2]{\mylist{#1}{#2}{k}}
\newcommand{\myresourcelist}[2]{\mylist{#1}{#2}{i}}

\newtheorem{theorem}{Theorem}
\newtheorem{lemma}[theorem]{Lemma}
\newtheorem{corollary}[theorem]{Corollary}
\newtheorem{definition}{Definition}
\newtheorem{example}{Example}

\hypersetup{
    colorlinks=false,
    pdfborder={0 0 0},
    pdftitle={Local approximation algorithms for a class of 0/1 max-min linear programs},
    pdfauthor={Patrik Floréen, Marja Hassinen, Petteri Kaski, Jukka Suomela},
}

\begin{document}

\title{Local Approximation Algorithms for a Class of \\ 0/1 Max-Min Linear Programs}

\author{Patrik~Floréen, Marja~Hassinen, Petteri~Kaski, and Jukka~Suomela \\
    Helsinki Institute for Information Technology HIIT, \\
    Department of Computer Science, University of Helsinki, \\
    P.O. Box 68, FI-00014 University of Helsinki, Finland\\
    $\{$firstname.lastname$\}$@cs.helsinki.fi
}

\maketitle

\begin{abstract}
    We study the applicability of distributed, local algorithms to 0/1 max-min LPs where the objective is to maximise ${\min_k \sum_v c_{kv} x_v}$ subject to ${\sum_v a_{iv} x_v \le 1}$ for each $i$ and ${x_v \ge 0}$ for each $v$. Here $c_{kv} \in \{0,1\}$, $a_{iv} \in \{0,1\}$, and the support sets ${V_i = \{ v : a_{iv} > 0 \}}$ and ${V_k = \{ v : c_{kv}>0 \}}$ have bounded size; in particular, we study the case $|V_k| \le 2$. Each agent $v$ is responsible for choosing the value of $x_v$ based on information within its constant-size neighbourhood; the communication network is the hypergraph where the sets $V_k$ and $V_i$ constitute the hyperedges. We present a local approximation algorithm which achieves an approximation ratio arbitrarily close to the theoretical lower bound presented in prior work.
\end{abstract}

\section{Introduction}

To motivate the problem setting studied in this paper, consider the toy network depicted in Fig.~\ref{fig:isp}. There are seven customers, $k_1,k_2,\ldots,k_7$, who are served by five access points, $i_1,i_2,\ldots,i_5$. The customers and access points are connected by the 14 numbered links.

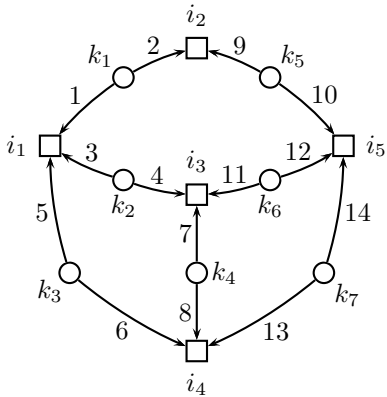
\begin{figure}[h]
    \centering
    \begin{pspicture}(0.5,0.5)(4.3,4.6)%
        {%
            \mypssources
            \Cnode(1.75, 3.8){k1}
            \Cnode(1.75, 2.75){k2}
            \Cnode(1.2, 1.8){k3}
            \Cnode(2.5, 1.8){k4}
            \Cnode(3.25, 3.8){k5}
            \Cnode(3.25, 2.75){k6}
            \Cnode(3.8, 1.8){k7}
            \uput{0.15}[135](k1){$k_1$}
            \uput{0.15}[d](k2){$k_2$}
            \uput{0.1}[225](k3){$k_3$}
            \uput{0.15}[r](k4){$k_4$}
            \uput{0.15}[45](k5){$k_5$}
            \uput{0.15}[d](k6){$k_6$}
            \uput{0.15}[315](k7){$k_7$}
        }{%
            \mypstargets
            \fnode(1.0, 3.1){i1}
            \fnode(2.5, 4.1){i2}
            \fnode(2.5, 2.6){i3}
            \fnode(2.5, 1.0){i4}
            \fnode(4.0, 3.1){i5}
            \uput[l](i1){$i_1$}
            \uput[u](i2){$i_2$}
            \uput[u](i3){$i_3$}
            \uput[d](i4){$i_4$}
            \uput[r](i5){$i_5$}
        }{%
            \mypslinks
            \newcommand{\mync}[4]{%
                \ncarc[arcangle=#4]{->}{#1}{#2}
                \naput{$\Myagents(#3)$}
            }
            \ncarc[arcangle=-8]{->}{k1}{i1}\nbput[npos=0.5]{$1$}
            \ncarc[arcangle=8]{->}{k1}{i2}\naput[npos=0.5]{$2$}
            \ncarc[arcangle=8]{->}{k2}{i1}\nbput[npos=0.5]{$3$}
            \ncarc[arcangle=-8]{->}{k2}{i3}\naput[npos=0.4]{$4$}
            \ncarc[arcangle=8]{->}{k3}{i1}\naput[npos=0.5]{$5$}
            \ncarc[arcangle=-8]{->}{k3}{i4}\nbput[npos=0.5]{$6$}
            \ncarc[arcangle=0]{->}{k4}{i3}\naput[npos=0.5]{$7$}
            \ncarc[arcangle=0]{->}{k4}{i4}\nbput[npos=0.5]{$8$}
            \ncarc[arcangle=-8]{->}{k5}{i2}\nbput[npos=0.5]{$9$}
            \ncarc[arcangle=8]{->}{k5}{i5}\naput[npos=0.5]{$10$}
            \ncarc[arcangle=8]{->}{k6}{i3}\nbput[npos=0.4]{$11$}
            \ncarc[arcangle=-8]{->}{k6}{i5}\naput[npos=0.5]{$12$}
            \ncarc[arcangle=8]{->}{k7}{i4}\naput[npos=0.5]{$13$}
            \ncarc[arcangle=-8]{->}{k7}{i5}\nbput[npos=0.5]{$14$}
        }
    \end{pspicture}
    \caption{An example of a data communication network.}\label{fig:isp}
\end{figure}

Now, suppose that we want to provide a maximum fair share of bandwidth to each customer, subject to the constraint that each access point can handle at most $1$~unit of network traffic. Put otherwise, we want to maximise the minimum bandwidth available to a customer.

In formally precise terms, we want to solve the following optimisation problem, where the variables $x_1,x_2,\ldots,x_{14}$ determine the amount of network traffic allocated to each link:
\begin{equation}
    \begin{alignedat}{2}
        &\text{maximise } \ && \myutil = \min \{\myparty{1}{x_1+x_2}, \myparty{2}{x_3+x_4}, \dotsc, \myparty{7}{x_{13}+x_{14}} \} \\
        &\text{subject to } \ && x_{1\phantom{1}} + x_{3\phantom{1}} + x_{5\phantom{1}} \le 1, \myresource{1} \\
        &&& x_{2\phantom{1}} + x_{9\phantom{1}} \phantom{{} + x_{00}}  \le 1, \myresource{2} \\
        &&& x_{4\phantom{1}} + x_{7\phantom{1}} + x_{11} \le 1, \myresource{3} \\
        &&& x_{6\phantom{1}} + x_{8\phantom{1}} + x_{13} \le 1, \myresource{4} \\
        &&& x_{10} + x_{12} + x_{14} \le 1, \myresource{5} \\
        &&& x_{1}, x_{2}, \dotsc, x_{14} \ge 0 .
    \end{alignedat}\label{eq:isp}
\end{equation}
An optimal solution is
$x_1 = x_7 = 2/7$,
$x_2 = x_8 = 3/7$,
$x_3 = x_6 = x_{11} = 0$,
$x_4 = x_5 = x_{12} = 5/7$,
$x_9 = x_{13} = 4/7$, and
$x_{10} = x_{14} = 1/7$,
guaranteeing the bandwidth $\myutil = 5/7$ to each customer. This is the best possible fair bandwidth allocation for our toy network. Moreover, it can be argued that such an allocation is not completely trivial to find with heuristic techniques, even in the toy network.

So far so good, but of course no one would seriously suggest a similar approach for optimising a real-world network. For one, any realistic network is several orders of magnitude larger, and, what is more, \emph{under constant change}. In particular, it is not feasible to put together a snapshot of the relevant topology of the entire network, such as Fig.~\ref{fig:isp}, for purposes of optimisation.

Nevertheless, a disciplined global optimisation approach, such as~\eqref{eq:isp}, provides an unequivocal benchmark for the design of distributed algorithms. Ideally, after each change in topology, the entire network should immediately converge to a global optimum. Of course, this ideal is unattainable if the nodes are only aware of their local neighbourhoods in the network. But not completely so: in certain cases local information does suffice to \emph{provably approximate} the global optimum.

In this work we present a novel distributed algorithm for linear max--min optimisation problems such as~\eqref{eq:isp}. The algorithm is both an approximation algorithm, with a provable approximation guarantee, and a local algorithm, with a constant local horizon~$r$ which is independent of the size of the network (see Section~\ref{ssec:local}). In practical terms, this implies all of the following.
\begin{itemize}
    \item The algorithm converges in $r$ time units and recovers from a topology change in $r$ time units.
    \item Whenever the network -- or any part of it -- has remained stable for $r$ time units, the algorithm provides a provable approximation guarantee for that part.
    \item A topology change only affects those parts of the network that are within $r$ hops from a node that loses or gains neighbours. The rest of the network stays in its current configuration, which is feasible and approximately optimal both before and after the topology change.
\end{itemize}

\subsection{Local Algorithms}\label{ssec:local}

We say that a distributed algorithm has the \emph{local horizon} $r$ if a topology change at node $v$ affects only those network nodes which are within $r$ hops from node $v$. In other words, the output of node $u$ is a function of input available in its radius $r$ neighbourhood.

Distributed algorithms where the local horizon $r$ is constant are called \emph{local algorithms} or distributed constant-time algorithms. Naturally the local setting is very restrictive; there are fundamental obstacles which prevent us from solving problems by using a local algorithm~\cite{linial92locality,kuhn04what}. However, a few positive results are known~\cite{papadimitriou93linear,naor95what,kuhn05price,kuhn05locality,kuhn06fault-tolerant,kuhn06price,kuhn05constant-time,floreen08approximating}. Our work presents a new example of such positive results.

If we assume that some auxiliary information -- such as the coordinates of the nodes -- is available, we can design local algorithms for a wider range of problems~\cite{czyzowicz08local,floreen07local,urrutia07local}. In the present work no such assumptions are necessary.

\subsection{Max-Min Packing Problem}

Formally, the problem setting that we study is a 0/1-version of the \emph{max-min packing problem}~\cite{floreen08approximating}, defined as follows. Let $V$, $I$ and $K$ be disjoint index sets; we say that each $v \in V$ is an \emph{agent}, each $k \in K$ is a \emph{beneficiary party}, and each $i \in I$ is a \emph{resource}. We assume that one unit of activity by $v$ benefits the party~$k$ by $c_{kv} \in \{0,1\}$ units and consumes $a_{iv} \in \{0,1\}$ units of the resource~$i$; the objective is to set the activities to provide a fair share of benefit for each party. Assuming that the activity of agent $v$ is $x_v$ units, the objective is to
\begin{equation}
    \begin{aligned}
        &\text{maximise } & \myutil = \min_{k \in K} \sum_{v \in V} c_{kv} x_v & \\
        &\text{subject to } & \sum_{v \in V} a_{iv} x_v &\le 1 & \quad\forall &i \in I, \\
        && x_v &\ge 0 & \forall &v \in V .
    \end{aligned}\label{eq:max-min}
\end{equation}
We assume that the support sets defined for all $i\in I$, $k\in K$, and $v\in V$ by
\begin{align*}
    V_i &= \{v\in V:a_{iv}>0\}, \\
    V_k &= \{v\in V:c_{kv}>0\}, \\
    I_v &= \{ i \in I : a_{iv}>0\}, \\
    K_v &= \{ k \in K : c_{kv}>0 \}
\end{align*}
have bounded size. That is, we focus on instances of~\eqref{eq:max-min} such that ${\mysize{I_v} \le \mybound IV}$, ${\mysize{K_v} \le \mybound KV}$, ${\mysize{V_i} \le \mybound VI}$ and ${\mysize{V_k} \le \mybound VK}$ for some constants $\mybound IV$, $\mybound KV$, $\mybound VI$ and $\mybound VK$. To avoid uninteresting degenerate cases, we assume that $I_v$, $V_i$ and $V_k$ are nonempty.

\begin{example}
    The problem instance~\eqref{eq:isp} is of the form~\eqref{eq:max-min}. There is one agent for each link. Customers $k_1, k_2, \dotsc, k_7$ are beneficiary parties and access points $i_1, i_2, \dotsc, i_5$ are resources. We have ${\mybound VK = 2}$ and ${\mybound VI = 3}$.
\end{example}

\subsection{Local, Distributed Setting}

The model of distributed computation assumed in this work is as follows. Each agent $v\in V$ is an independent computational entity; all agents execute the same deterministic algorithm. Agent $v$ controls the associated variable~$x_v$.

The communication between the agents is constrained by the \emph{communication graph}, a hypergraph $\myH = (V, \myE)$, where the vertices $V$ are the agents and the hyperedges are defined by $\myE = {\{ V_i : i \in I \}} \, \cup \, {\{ V_k : k \in K \}}$. Two agents can communicate directly with each other if they are adjacent in $\myH$. Let $d_\myH(u,v)$ be the shortest-path distance (number of hyperedges, hop count) between $u \in V$ and $v \in V$ in $\myH$, and let
\[
    B_\myH(v, r) = \{u\in V:d_\myH(u,v)\leq r\},\quad  r=1,2,\ldots,
\]
be the set of nodes within distance at most $r$ from node $v$ in hypergraph~$\myH$. 

Each agent $v\in V$ has the following local input: the identifier of $v$, the hyperedges $V_i$ for which $v \in V_i$, and the hyperedges $V_k$ for which $v \in V_k$. The hyperedges are given as sets of identifiers.

The algorithm executed by the agents has the \emph{local horizon} $r$ if, for every agent $v\in V$, the value set to $x_v$ is a function of the local input of the agents in $B_\myH(v, r)$. 

Thus, each agent $v$ executing an algorithm with a local horizon $r$ is completely oblivious to the network beyond $B_\myH(v,r+1)$. In particular, two distinct agents $u,v\in V$ may have the same identifier if $d_\myH(u,v)\geq 2r+3$. Thus, without loss of generality we may assume that the local input of each agent has a size (in bits) that depends only on $r$, $\mybound IV$, $\mybound KV$, $\mybound VI$, and $\mybound VK$, but not on the size of the network.

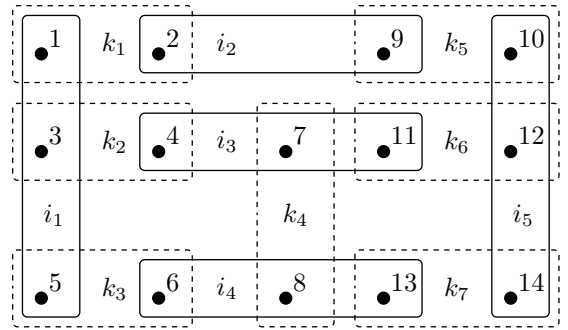
\begin{figure}[b]
    \centering
    \begin{pspicture}(0.5,0.7)(6.3,4.0)%
        {%
            \mypsihyperedges
            \myframeb{1.0}{1.0}{1.0}{3.5}{0.20}{0.40}
            \myframeb{2.2}{3.5}{4.5}{3.5}{0.20}{0.40}
            \myframeb{2.2}{2.5}{4.5}{2.5}{0.20}{0.40}
            \myframeb{2.2}{1.0}{4.5}{1.0}{0.20}{0.40}
            \myframeb{5.8}{1.0}{5.8}{3.5}{0.20}{0.40}
            \rput(1.13,1.85){$i_1$}
            \rput(2.9,3.6){$i_2$}
            \rput(2.9,2.6){$i_3$}
            \rput(2.9,1.1){$i_4$}
            \rput(5.93,1.85){$i_5$}
        }{%
            \mypskhyperedges
            \myframe{1.0}{3.5}{2.1}{3.5}{0.30}{0.30}{0.45}{0.50}
            \myframe{1.0}{2.5}{2.1}{2.5}{0.30}{0.30}{0.45}{0.50}
            \myframe{1.0}{1.0}{2.1}{1.0}{0.30}{0.30}{0.45}{0.50}
            \myframeb{3.5}{1.0}{3.5}{2.5}{0.30}{0.50}
            \myframeb{4.5}{3.5}{5.8}{3.5}{0.30}{0.50}
            \myframeb{4.5}{2.5}{5.8}{2.5}{0.30}{0.50}
            \myframeb{4.5}{1.0}{5.8}{1.0}{0.30}{0.50}
            \rput(1.75,3.6){$k_1$}
            \rput(1.75,2.6){$k_2$}
            \rput(1.75,1.1){$k_3$}
            \rput(3.6,1.85){$k_4$}
            \rput(5.25,3.6){$k_5$}
            \rput(5.25,2.6){$k_6$}
            \rput(5.25,1.1){$k_7$}
        }{%
            \mypshypernodes
            \newcommand{\mynode}[3]{%
                \Cnode*(#1,#2){v#3}
                \uput[ur](v#3){$#3$}
            }%
            \mynode{1}{3.5}{1}
            \mynode{1}{2.5}{3}
            \mynode{1}{1}{5}
            \mynode{2.2}{3.5}{2}
            \mynode{2.2}{2.5}{4}
            \mynode{2.2}{1}{6}
            \mynode{3.5}{2.5}{7}
            \mynode{3.5}{1}{8}
            \mynode{4.5}{3.5}{9}
            \mynode{4.5}{2.5}{11}
            \mynode{4.5}{1}{13}
            \mynode{5.8}{3.5}{10}
            \mynode{5.8}{2.5}{12}
            \mynode{5.8}{1}{14}
        }
    \end{pspicture}
    \caption{Hypergraph $\myH$ for instance~\eqref{eq:isp}.}\label{fig:isp-hyper}
\end{figure}

\begin{example}
    Fig.~\ref{fig:isp-hyper} shows the hypergraph $\myH$ for the problem instance~\eqref{eq:isp}. The structure of $\myH$ closely reflects the structure of the network shown in Fig.~\ref{fig:isp}: two agents are able to communicate with each other if they share the same access point or the same customer.
\end{example}

\subsection{Prior Work and Contributions}

This paper contributes to work in progress aimed at a complete characterisation of the local approximability of the max-min packing problem. Here we provide the answer for the case of 0/1 coefficients and $\mybound VK = 2$:

\begin{theorem}\label{thm:upper-bound}
    Let $\mybound VI \ge 2$ be given. For any $\epsilon > 0$, there is a local algorithm for 0/1 max-min packing problem~\eqref{eq:max-min} with the approximation ratio $\mybound VI/2 + \epsilon$, assuming $\mybound VK = 2$.
\end{theorem}

This upper bound is tight; by prior work~\cite[Corollary~2]{floreen08approximating} we know that for a given $\mybound VI \ge 2$, there is no local approxi\-mation algorithm for~\eqref{eq:max-min} with an approximation ratio less than $\mybound VI/2$, and this holds even if $\mybound VK = 2$.

The \emph{safe algorithm}~\cite{papadimitriou93linear,floreen08approximating} achieves the approximation ratio $\mybound VI$ for~\eqref{eq:max-min}. Our algorithm improves this by a factor~of~$2$.

The proof of Theorem~\ref{thm:upper-bound} is structured as follows. First, Section~\ref{sec:reduce} presents a simple modification of~\eqref{eq:max-min} which reduces the size of each constraint to~$2$, that is, we arrive at an instance with ${\mybound VI = \mybound VK = 2}$.

The rest of this work, starting from Section~\ref{sec:presentation}, presents a \emph{local approximation scheme} for the special case ${\mybound VI = \mybound VK = 2}$. A local approximation scheme is a family of local algorithms such that for any $\epsilon > 0$ there is a local algorithm which achieves the approximation ratio $1 + \epsilon$.

The local approximation scheme and the reduction of Section~\ref{sec:reduce} constitute the proof of Theorem~\ref{thm:upper-bound}. We are able to achieve an approximation ratio arbitrarily close to the lower bound $\mybound VI/2$, in spite of the crude reduction that we used in Section~\ref{sec:reduce} to bring $\mybound VI$ down to~$2$.

\section{Reducing the Size of Constraints}\label{sec:reduce}

We first wishfully assume that for any $\epsilon' > 0$ there is a local approximation algorithm which achieves the approximation ratio $1 + \epsilon'$ for the special case ${\mybound VI = \mybound VK = 2}$. In this section, we show that this assumption directly implies our main result, Theorem~\ref{thm:upper-bound}.

Fix an $\epsilon > 0$ and a bound $\mybound VI > 2$. Given an instance of~\eqref{eq:max-min}, we replace each constraint which involves more than $2$ variables by several constraints which involve exactly $2$ variables each. In precise terms, consider $i \in I$ such that $\mysize{V_i} > 2$. Let $n = \mysize{V_i}$. Remove constraint $i$ from the instance. Add $\binom{n}{2}$ distinct constraints $x_u + x_v \le 1$ where $u, v \in V_i$, $u \ne v$. For example, the constraint $x_1 + x_2 + x_3 \le 1$ is replaced by the set of constraints $x_1 + x_2 \le 1$, $x_1 + x_3 \le 1$, and $x_2 + x_3 \le 1$. This can be done by a local algorithm.

The set of feasible solutions differs between the modified instance and the original instance. However, the utility of a solution, $\myutil(x) = \min_k \sum_v c_{kv} x_v$, is the same in both instances.

Once we have constructed the modified instance, we apply the local approximation scheme to solve it within the approximation ratio $1+\epsilon'$ where $\epsilon' = 2 \epsilon / \mybound VI$; let $x'$ be the solution. We form a solution $x$ of the original instance by setting $x_v = 2 x'_v / \mybound VI$.

First, we show that $x$ is a feasible solution of the original instance. As $x'$ satisfies all constraints of size at most $2$, so does $x$. Now consider a constraint $i$ in the original instance with more than $2$ variables. Add up all new constraints which replace $i$ in the modified instance to obtain
\[
    {(n-1) x'_1} + {(n-1) x'_2} + \dotsb + {(n-1) x'_n} \le \tbinom{n}{2}
\]
which implies $x_1 + x_2 + \dotsb + x_n \le n / \mybound VI \le 1$.

Second, we show that $x$ is a $(\mybound VI/2 + \epsilon)$-approximate solution of the original instance. Let $x^{*}$ be an optimal solution of the original instance. Now $x^{*}$ is also a feasible solution of the modified instance, and $\myutil(x^{*})$ is a lower bound for the optimum value of the modified instance. Therefore $\myutil(x') \ge \myutil(x^{*}) / (1+\epsilon')$. By the choice of $x$, we conclude that
\[
    \myutil(x) = \frac{2 \myutil(x')}{\mybound VI} \ge \frac{\myutil(x^{*})}{\mybound VI/2 + \epsilon}.
\]
\vspace{0ex}

\begin{figure*}
    \centering
    \begin{pspicture}(-1.5,0.9)(8.0,5.5)%
        \newlength{\myoriga}\pssetlength{\myoriga}{0.0}%
        \newlength{\myorigb}\pssetlength{\myorigb}{2.5}%
        \newlength{\myorigc}\pssetlength{\myorigc}{5.0}%
        \newlength{\myorigd}\pssetlength{\myorigd}{7.5}%
        \newlength{\myonea}\pssetlength{\myonea}{\myoriga}%
        \newlength{\myoneb}\pssetlength{\myoneb}{\myorigb}%
        \newlength{\myonec}\pssetlength{\myonec}{\myorigc}%
        \newlength{\myoned}\pssetlength{\myoned}{\myorigd}%
        \psaddtolength{\myonea}{-1.0}%
        \psaddtolength{\myoneb}{-1.0}%
        \psaddtolength{\myonec}{-1.0}%
        \psaddtolength{\myoned}{-1.0}%
        \rput(\myonea,1.0){(a)}
        \rput(\myoneb,1.0){(b)}
        \rput(\myonec,1.0){(c)}
        \rput(\myoned,1.0){(d)}
        {%
            \mypsihyperedges
            \myframeb{\myoriga}{1}{\myoriga}{2}{0.15}{0.35}
            \myframeb{\myoriga}{2}{\myoriga}{3}{0.25}{0.45}
            \myframeb{\myonea}{4}{\myoriga}{4}{0.25}{0.45}
            \myframebL{\myonea}{3}{\myoriga}{4}{0.15}{0.35}
            \myframeb{\myorigb}{3}{\myorigb}{3}{0.25}{0.45}
            \myframeb{\myoneb}{4}{\myorigb}{4}{0.25}{0.45}
            \myframebL{\myoneb}{3}{\myorigb}{4}{0.15}{0.35}
            \myframeb{\myorigc}{2}{\myorigc}{3}{0.25}{0.45}
            \myframeb{\myonec}{4}{\myorigc}{4}{0.25}{0.45}
            \myframebL{\myonec}{3}{\myorigc}{4}{0.15}{0.35}
        }{%
            \mypskhyperedges
            \myframe{\myoriga}{3}{\myoriga}{4}{0.35}{0.25}{0.55}{0.55}
            \myframe{\myonea}{4}{\myonea}{4}{0.35}{0.35}{0.45}{0.55}
            \myframe{\myorigb}{3}{\myorigb}{4}{0.35}{0.35}{0.55}{0.55}
            \myframe{\myoneb}{4}{\myoneb}{4}{0.35}{0.35}{0.45}{0.55}
            \myframe{\myorigc}{3}{\myorigc}{4}{0.35}{0.25}{0.55}{0.55}
            \myframe{\myonec}{4}{\myonec}{5}{0.35}{0.35}{0.45}{0.45}
        }{%
            \mypshypernodes
            \newcommand{\mynode}[3]{%
                \Cnode*(#2,#1){dummy}
                \uput[ur](dummy){$#3$}
            }%
            \newcommand{\mynodez}[3]{%
                \Cnode(#2,#1){dummy}
                \uput[ur](dummy){$#3$}
            }%
            \mynode{1}{\myoriga}{5}
            \mynode{2}{\myoriga}{4}
            \mynode{3}{\myoriga}{3}
            \mynode{4}{\myoriga}{2}
            \mynode{4}{\myonea}{1}
            \mynode{3}{\myorigb}{3}
            \mynode{4}{\myorigb}{2}
            \mynode{4}{\myoneb}{1}
            \mynodez{2}{\myorigc}{6}
            \mynode{3}{\myorigc}{3}
            \mynode{4}{\myorigc}{2}
            \mynode{4}{\myonec}{1}
            \mynodez{5}{\myonec}{7}
        }{%
            \mypsnodes
            \newcommand{\mynode}[5]{%
                \Cnode*(#2,#1){#4}
                \uput[#5](#4){$#3$}
            }%
            \newcommand{\mynodez}[5]{%
                \Cnode(#2,#1){#4}
                \uput[#5](#4){$#3$}
            }%
            \mynodez{2}{\myorigd}{6}{v6}{r}
            \mynode{3}{\myorigd}{3}{v3}{r}
            \mynode{4}{\myorigd}{2}{v2}{r}
            \mynode{4}{\myoned}{1}{v1}{l}
            \mynodez{5}{\myoned}{7}{v7}{l}
        }{%
            \mypsiedges
            \ncarc{v1}{v2}
            \naput{$i_1$}
            \ncarc{v1}{v3}
            \nbput{$i_2$}
            \ncarc{v6}{v3}
            \nbput{$i_3$}
        }{%
            \mypskedges
            \ncarc{v1}{v7}
            \naput{$k_1$}
            \ncarc{v2}{v3}
            \naput{$k_2$}
        }
    \end{pspicture}
    \caption{Transforming the problem instance.}\label{fig:prelim}
\end{figure*}
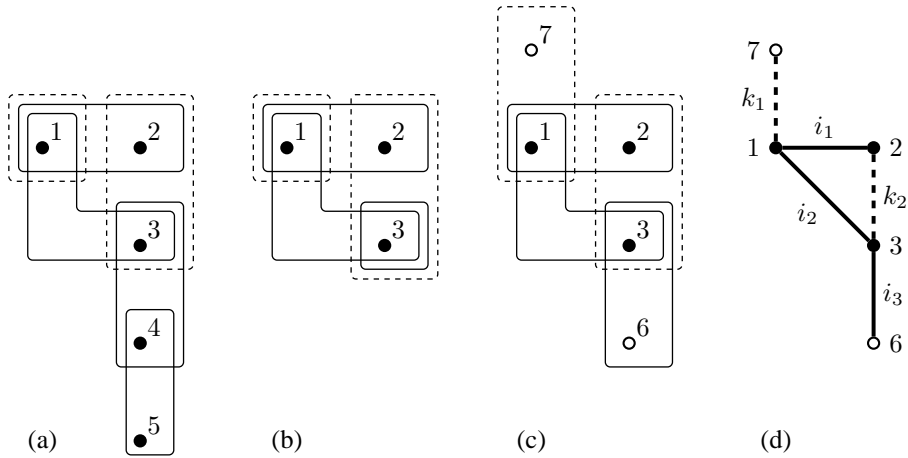

\section{Presentation as a Graph}\label{sec:presentation}

We proceed to show that there indeed is a local approximation scheme for the special case ${\mybound VI = \mybound VK = 2}$. To simplify the discussion, we present the problem instance as an undirected multigraph $\myG$, where both edges and vertices are two-coloured. This allows us to describe the algorithm in graph-theoretic terms.

\begin{example}
    We use the following instance of~\eqref{eq:max-min} to illustrate the presentation as a graph. The beneficiary parties are $K = \{\mypartylist{2}{,}\}$ and the constraints are $I = \{\myresourcelist{4}{,}\}$. The objective is to
    \begin{equation}
        \begin{aligned}
            &\text{maximise } && \myutil = \min {\{ \myparty{1}{x_1}, \myparty{2}{x_2+x_3} \}} \\
            &\text{subject to } && x_1 + x_2 \le 1, \myresource{1} \\
            &&& x_1 + x_3 \le 1, \myresource{2} \\
            &&& x_3 + x_4 \le 1, \myresource{3} \\
            &&& x_4 + x_5 \le 1, \myresource{4} \\
            &&& x_1, x_2, \dotsc, x_5 \ge 0 .
        \end{aligned}\label{eq:ex:prelim}
    \end{equation}
    The hypergraph~$\myH$ is illustrated in Fig.~\ref{fig:prelim}a; solid lines are hyperedges~$V_i$ and dashed lines are hyperedges~$V_k$. An optimal solution with $\myutil = 2/3$ is $x_1 = 2/3$, $x_2 = x_3 = 1/3$, and $x_4 = x_5 = 0$.
\end{example}

\subsection{Remove Non-Contributing Agents}

We have assumed that $I_v$, $V_i$ and $V_k$ are nonempty for each $v \in V$, $i \in I$ and $k \in K$. We can make a further assumption that $K_v$ is nonempty for each $v$. If this is not the case for some $v$, we can simply choose $x_v = 0$ and remove the agent $v$ from the problem instance. If such changes make $V_i$ empty for some $i$, we can remove the redundant constraint $i$. These modifications can be done by a local algorithm; this is step illustrated in Fig.~\ref{fig:prelim}b.

\subsection{Hyperedges of Size 2 Only}

At this point, $\mysize{V_k} \in \{1,2\}$ for each $k \in K$ and $\mysize{V_i} \in \{1,2\}$ for each $i \in I$. If $\mysize{V_k} = 1$ for some $k$, we add a new agent $v$ into $V$. The variable $x_v$ controlled by agent $v$ is forced to $0$ by adding the constraint $x_v = 0$. Now we can set $c_{kv} = 1$ without changing the solution. Similarly, if $\mysize{V_i} = 1$ for some $i$, we add a new agent $v$ into $V$, we force $x_v = 0$, and we set $a_{iv} = 1$.

After these changes, $\mysize{V_k} = 2$ for each $k \in K$ and $\mysize{V_i} = 2$ for each $i \in I$. This simple structure comes at the cost of having some new agents $v$ for which we force $x_v = 0$; we also allow $K_v = \emptyset$ or $I_v = \emptyset$ for such agents.

\begin{example}
    This step is illustrated in Fig.~\ref{fig:prelim}c. We have transformed~\eqref{eq:ex:prelim} into the following form: maximise $\myutil = \min \{ x_1 + x_6, x_2+x_3 \}$ subject to $x_1 + x_2 \le 1$, $x_1 + x_3 \le 1$, $x_3 + x_7 \le 1$, $x_1, x_2, x_3 \ge 0$, and $x_6 = x_7 = 0$.
\end{example}

\subsection{Construct the Graph}

Next we represent the modified problem instance as an undirected multigraph~$\myG$. The set of vertices of $\myG$ is the set of agents $V$; the vertices $v$ for which we force $x_v = 0$ are called $0$-vertices and the remaining vertices are called $x$-vertices. For each party $k \in K$, we have the edge $V_k$; these are called \Kedges{}. For each constraint $i \in I$, we have the edge $V_i$; these are called \Iedges{}. There are no other edges. If there is an \Iedge{} $\{u,v\}$, we say that $u$ and $v$ are \myIdash\emph{adjacent}. We define \myKdash\emph{adjacent} vertices analogously.

In other words, the vertices of $\myG$ are coloured with two colours, $0$ and $x$, and the edges are coloured with two colours, $K$ and $I$. We have encoded the original problem instance as a coloured graph $\myG$.

The graph $\myG$ for the instance~\eqref{eq:ex:prelim} is illustrated in Fig.~\ref{fig:prelim}d. Open circles are $0$-vertices and closed circles are $x$-vertices; solid lines are \Iedges{} and dashed lines are \Kedges{}.

\section{Definitions}

\begin{definition}
Let $X, Y \in \{K, I\}$. A $(v_0, X, Y, v_n)$-\emph{walk} is a finite sequence of the form $(v_0, e_1, v_1, e_2, v_2, \dotsc, e_n, v_n)$ which satisfies all of the following: each $v_j$ is a vertex of $\myG$; each $e_j$ is an edge of the form $\{v_{j-1}, v_j\}$ in the graph $\myG$; the edges $e_j$ are alternately \Kedges{} and \Iedges{}; $e_1$ is an \Xedge{}; and $e_n$ is a \Yedge{}.
A $(v,X,Y,0)$-\emph{walk} is a $(v,X,Y,u)$-walk where $u$ is a $0$-vertex.
A $(0,X,Y,0)$-\emph{walk} is a $(v,X,Y,u)$-walk where $v$ and $u$ are a $0$-vertices.
The \myKdash\emph{length} of a walk is the number of \Kedges{} in the walk.
\end{definition}

We emphasise that (i)~there can be repeated edges and repeated vertices in walks; and (ii)~all walks throughout this work are \emph{alternating} walks where \Kedges{} and \Iedges{} alternate.

\begin{definition}
Let $v$ be an $x$-vertex and let ${X, Y \in \{K, I\}}$. We write $a(v,X,Y,0)$ for the minimum \Klength{} of a $(v,X,Y,0)$-walk; if no $(v,X,Y,0)$-walk exists, then we define that $a(v,X,Y,0) = \infty$. We write $A(v,X)$ for the maximum \Klength{} of a $(v,X,\cdot,\cdot)$-walk; if such walks with an arbitrarily large \Klength{} exist, then we define that ${A(v,X) = \infty}$.
\end{definition}

\begin{example}
Consider the vertex $1 \in V$ in Fig.~\ref{fig:prelim}d. We have $a(1,K,K,0) = 1$, $a(1,K,I,0) = \infty$, $a(1,I,K,0) = 2$, $a(1,I,I,0) = 1$, $A(1,K) = 1$ and $A(1,I) = 2$. Note that it is possible to have $A(1,K) < a(1,K,I,0)$.
\end{example}

Fix a constant $R \in \{1,2,\dotsc\}$. We define the bounded versions of $a$ and $A$ by
\begin{align*}
    b(v,X,Y,0) &= \min \{ a(v,X,Y,0), R \}, \\
    B(v,X) &= \min \{ A(v,X), R \}
\end{align*}
for each $X, Y \in \{K, I\}$.

\section{Local Algorithm}

Now we are ready to present the local approximation al\-go\-rithm. More accurately, we present a local approximation scheme, a family of algorithms parametrised by the constant~$R$. The value of $R$ determines the desired trade-off between the local horizon and the approximation ratio: the local horizon is $2R$ and the approximation ratio is $1 + 1/(R-1)$.

A local algorithm with any finite local horizon cannot determine the value of $a(v,X,Y,0)$ or $A(v,X)$ in the general case. However, assuming that the local horizon is $2R$, then each agent $v$ can determine locally whether $a(v,X,K,0) \le R$ or not. Furthermore, each agent $v$ can determine locally the value of $b(v,X,K,0)$ and $B(v,X)$. It turns out that this information is sufficient in order to obtain an approximation algorithm.

Our local algorithm consists of two steps. In the first step, each $x$-vertex $v$ determines whether $a(v,K,K,0) \le R$, whether $a(v,I,K,0) \le R$, and what are the values of $b(v,I,K,0)$, $b(v,K,K,0)$, $B(v,I)$, and $B(v,K)$. To implement this step in a real-world distributed system, it is sufficient to propagate \Khop{} counters along alternating walks in $\myG$ for $2R$ communication rounds.

\begin{subequations}
In the second step, each $x$-vertex $v$ performs the following local computations. First, choose the value $p_v$ as follows.
\begin{alignat}{2}
    p_v &= b(v,I,K,0) & \text{if } a(v,K,K,0) \le R,& \label{pa} \\
    p_v &= \min {\{ b(v,I,K,0), B(v,K) \}} \mspace{-30mu} & \text{otherwise}.& \label{pb}
\end{alignat}
\end{subequations}
\begin{subequations}
Choose the value $q_v$ in an analogous manner.
\begin{alignat}{2}
    q_v &= b(v,K,K,0) & \text{if } a(v,I,K,0) \le R,& \label{qa} \\
    q_v &= \min {\{ b(v,K,K,0), B(v,I) \}} \mspace{-20mu} & \text{otherwise}.& \label{qb}
\end{alignat}
\end{subequations}
Finally, let $x_v = p_v / (p_v + q_v)$. This value is the output of the vertex $v$.

\begin{example}
In Fig.~\ref{fig:prelim}d, agent $1 \in V$ chooses $x_1 = 1/2$ if $R = 1$ and $x_1 = 2/3$ if $R \ge 2$.
\end{example}

We now proceed to show that the chosen values $x_v$ provide a feasible and near-optimal solution to~\eqref{eq:max-min}.

\section{Auxiliary Results}

We begin with some observations on the structure of $\myG$. First, each $0$-vertex is incident to exactly one edge. Second, each $x$-vertex is incident to at least one \Kedge{} \emph{and} at least one \Iedge{}. Given an $x$-vertex $v$, we can construct both a $(v,K,\cdot,\cdot)$-walk and a $(v,I,\cdot,\cdot)$-walk with at least one edge, and we can extend such alternating walks indefinitely until we meet a $0$-vertex.

\begin{lemma}\label{lem:p-q-x-bound}
    For any $x$-vertex $v$, the local algorithm chooses $p_v \ge 1$, $q_v \ge 0$, and $x_v \le 1$.
\end{lemma}
\begin{proof}
    Follows from the definitions.
\end{proof}

\subsection{Bounds for the Optimum}

Now we give upper bounds for the optimum value of~\eqref{eq:max-min}. Let $\myoptx$ be an optimal solution and let $\myoptutil$ be its objective value.

\begin{lemma}\label{lem:walk}
    If there exists a $(v,I,K,u)$-walk of \Klength{} $n$, then $\myoptx_v - \myoptx_u \le (1-\myoptutil) n$.
\end{lemma}
\begin{proof}
    If $n=1$, then there is a vertex $t$, an \Iedge{} $\{v,t\}$, and a \Kedge{} $\{t,u\}$. Then $\myoptx_v + \myoptx_t \le 1$ and $\myoptx_t + \myoptx_u \ge \myoptutil$, that is, $\myoptx_v - \myoptx_u \le 1 - \myoptutil$. The claim follows by induction.
\end{proof}

\begin{corollary}\label{cor:long-walk}
    If there exists a $(v,I,K,u)$-walk of \Klength{} $n$, then $\myoptutil \le 1 + 1/n$.
\end{corollary}
\begin{proof}
    Follows from $\myoptx_u \le 1$, $\myoptx_v \ge 0$, and the previous lemma.
\end{proof}

\begin{corollary}\label{cor:short-0KK0-walk}
    If there exists a $(0,K,K,0)$-walk of \Klength{} $n$, then $\myoptutil \le 1 - 1/n$.
\end{corollary}
\begin{proof}
    The case $n = 1$ is not possible so assume ${n > 1}$. Then there is a $(v,I,K,u)$-walk of \Klength{} $n-1$ such that $u$ is a $0$-vertex and there is a \Kedge{} between $v$ and a $0$-vertex. Therefore $\myoptx_u = 0$ and $\myoptx_v \ge \myoptutil$. By Lemma~\ref{lem:walk},
    \begin{align*}
        {(1-\myoptutil) (n-1)} \ge {\myoptx_v - \myoptx_u} = {\myoptx_v \ge \myoptutil}.
    \end{align*}
    The claim follows.
\end{proof}

\begin{example}
    By Corollary~\ref{cor:short-0KK0-walk}, $\myoptutil = 2/3$ in~\eqref{eq:ex:prelim}.
\end{example}

\subsection{Adjacent Vertices}

\begin{lemma}\label{lem:I-adj}
    If $v$ and $u$ are \Iadjacent{} $x$-vertices, then
    \begin{align*}
        a(v,I,K,0) &\le a(u,K,K,0), \\
        b(v,I,K,0) &\le b(u,K,K,0), \\
        A(v,I) &\ge A(u,K), \\
        B(v,I) &\ge B(u,K).
    \end{align*}
\end{lemma}
\begin{proof}
    Any given $(u,K,Y,b)$-walk can be extended into a $(v,I,Y,b)$-walk by first taking the \Iedge{} $\{v,u\}$. The \Klength{} does not change.
\end{proof}
\begin{lemma}\label{lem:K-adj}
    If $v$ and $u$ are \Kadjacent{} $x$-vertices, then
    \begin{align*}
        a(v,K,K,0) &\le a(u,I,K,0) + 1, \\
        b(v,K,K,0) &\le b(u,I,K,0) + 1, \\
        A(v,K) &\ge A(u,I) + 1, \\
        B(v,K) &\ge B(u,I).
    \end{align*}
\end{lemma}
\begin{proof}
    Any given $(u,I,Y,b)$-walk can be extended into a $(v,K,Y,b)$-walk by first taking the \Kedge{} $\{v,u\}$. The \Klength{} increases by~$1$.
\end{proof}

\section{Feasibility}\label{sec:feas}

We show that the values $x_v$ chosen by the local algorithm provide a feasible solution to \eqref{eq:max-min}. Consider an \Iedge{} $\{v,u\}$. We need to prove that $x_v + x_u \le 1$. If $v$ or $u$ is a $0$-vertex, then the claim holds by Lemma~\ref{lem:p-q-x-bound}; we focus on the case that $v$ and $u$ are $x$-vertices. We begin with the following lemma.

\begin{lemma}\label{lem:pv-qu}
    If $v$ and $u$ are \Iadjacent{} $x$-vertices, then we have ${p_v \le q_u}$.
\end{lemma}
\begin{proof}
    First, assume that $a(v,K,K,0) \le R$. In this case, Lemma~\ref{lem:I-adj} implies that $a(u,I,K,0) \le R$. We have $p_v = b(v,I,K,0)$ and $q_u = b(u,K,K,0)$. We apply Lemma~\ref{lem:I-adj} again to obtain $p_v \le q_u$.

    Second, assume that $a(v,K,K,0) > R$. In this case, Lemma~\ref{lem:I-adj} implies that $b(v,I,K,0) \le b(u,K,K,0)$ and $B(v,K) \le B(u,I)$. We obtain
    \begin{align*}
        p_v &= \min {\{ b(v,I,K,0), B(v,K) \}} \\
            &\le \min {\{ b(u,K,K,0), B(u,I) \}} \le q_u .
    \end{align*}
    We conclude that the claim holds in both cases.
\end{proof}

\begin{corollary}
    If $v$ and $u$ are \Iadjacent{} $x$-vertices, then ${x_v + x_u \le 1}$.
\end{corollary}
\begin{proof}
    By Lemma~\ref{lem:pv-qu}, we have ${p_v \le q_u}$, and by symmetry, ${p_u \le q_v}$. Therefore
    \begin{align*}
    x_v + x_u &= \frac{p_v}{p_v + q_v} + \frac{p_u}{p_u + q_u} \\
        &\le \frac{p_v}{p_v + p_u} + \frac{p_u}{p_u + p_v} = 1 .
    \end{align*}
    This completes the proof.
\end{proof}

\section{Approximation Guarantee}\label{sec:opt}

Next we show that the values $x_v$ chosen by the local algorithm provide a near-optimal solution to \eqref{eq:max-min}. Consider a \Kedge{} $\{v,u\}$. We show that $x_v + x_u \ge \alpha \myoptutil$ where $\alpha = 1-1/R$.

\subsection{One \texorpdfstring{$x$}{x}-vertex and One 0-vertex}

Let us first consider the case where $v$ is an $x$-vertex and $u$ is a $0$-vertex. Then we have
\[
    q_v \le b(v,K,K,0) = a(v,K,K,0) = 1 \le R
\]
and $p_v = b(v,I,K,0)$.

\begin{lemma}\label{lem:opt-x0-va}
    If $a(v,I,K,0) \le R$, then $x_v + x_u \ge \myoptutil$.
\end{lemma}
\begin{proof}
    We have $q_v = 1$ and
    \[
        x_v + x_u  = x_v = {1 - 1/n}
    \]
    where $n = p_v+1$. There exists a $(0,K,K,0)$-walk of \Klength{} $n$, starting from $u$ and going through~$v$. Corollary~\ref{cor:short-0KK0-walk} implies $\myoptutil \le 1 - 1/n$.
\end{proof}

\begin{lemma}\label{lem:opt-x0-vb}
    If $a(v,I,K,0) > R$, then $x_v + x_u \ge \alpha \myoptutil$.
\end{lemma}
\begin{proof}
    We have $q_v \le 1$, $p_v = R$ and
    \[
        x_v + x_u = x_v \ge 1 - 1/R = \alpha.
    \]
    In the optimal solution, $\myoptx_v \le 1$ and $\myoptx_u = 0$. Therefore $\myoptutil \le 1$.
\end{proof}

\begin{corollary}\label{cor:opt-x0}
    If $x$-vertex $v$ and $0$-vertex $u$ are \Kadjacent{}, then $x_v + x_u \ge \alpha\myoptutil$.
\end{corollary}
\begin{proof}
    Apply Lemmata~\ref{lem:opt-x0-va} and \ref{lem:opt-x0-vb}.
\end{proof}

\subsection{Two \texorpdfstring{$x$}{x}-vertices}

Second, we consider the case where both $v$ and $u$ are $x$-vertices. There are several subcases to study.
\begin{lemma}\label{lem:opt-vpa-vqa}
    If $a(v,K,K,0) \le R$ and $a(v,I,K,0) \le R$, then $x_v + x_u \ge \myoptutil$.
\end{lemma}
\begin{proof}
    Regardless of whether $q_u$ satisfies~\eqref{qa} or~\eqref{qb}, by Lemma~\ref{lem:K-adj}
    \[
        q_u \le b(u,K,K,0) \le b(v,I,K,0) + 1 = p_v + 1.
    \]
    If $p_u$ satisfies~\eqref{pa}, we have
    \[
        p_u = b(u,I,K,0) \ge b(v,K,K,0) - 1 = q_v - 1.
    \]
    Otherwise $p_u$ satisfies~\eqref{pb}. We have $R < a(u,K,K,0) \le a(v,I,K,0) + 1 \le R + 1$, that is, $a(u,K,K,0) = R+1$. This implies $A(u,K) \ge R+1$, $B(u,K) = R$ and
    \begin{align*}
        p_u &= \min {\{ b(u,I,K,0), R \}} = b(u,I,K,0) \\
            &\ge b(v,K,K,0) - 1 = q_v - 1.
    \end{align*}
    In both cases we have $p_u \ge q_v - 1$. Therefore
    \begin{align*}
        x_v + x_u
            &\ge \frac{p_v}{p_v + q_v} + \frac{q_v - 1}{(q_v - 1) + (p_v + 1)} \\
            &= 1 - \frac{1}{p_v + q_v}.
    \end{align*}
    As there exists a $(0,K,K,0)$-walk of \Klength{} $p_v + q_v$, Corollary~\ref{cor:short-0KK0-walk} implies that $\myoptutil \le x_v + x_u$.
\end{proof}

\begin{lemma}\label{lem:opt-vpa-vqb}
    If $a(v,K,K,0) \le R$ and $a(v,I,K,0) > R$, then $x_v + x_u \ge \alpha\myoptutil$.
\end{lemma}
\begin{proof}
    Regardless of whether $q_u$ satisfies~\eqref{qa} or~\eqref{qb}, we have
    \[
        q_u \le R = b(v,I,K,0) = p_v.
    \]
    As for $p_u$, there are three cases. First, if $p_u$ satisfies~\eqref{pa}, we have by Lemma~\ref{lem:K-adj}
    \[
        p_u = b(u,I,K,0) \ge b(v,K,K,0) - 1 \ge q_v - 1.
    \]
    Second, if $p_u$ satisfies~\eqref{pb} and $b(u,I,K,0) < B(u,K)$, we have
    \[
        p_u = b(u,I,K,0) \ge b(v,K,K,0) - 1 \ge q_v - 1.
    \]
    Third, if $p_u$ satisfies~\eqref{pb} and $b(u,I,K,0) \ge B(u,K)$, we have
    \[
        p_u = B(u,K) \ge B(v,I) \ge q_v > q_v - 1.
    \]
    In each case we have $p_u \ge q_v - 1$. Therefore
    \begin{align*}
        x_v + x_u
            &= \frac{R}{R + q_v} + \frac{p_u}{p_u + q_u} \\
            &\ge \frac{R}{R + p_u + 1} + \frac{p_u}{p_u + R}
            \ge 1 - \frac{1}{R}
            = \alpha.
    \end{align*}
    Because $a(v,K,K,0) \le R$, there exists a $0$-vertex incident to a \Kedge{} and therefore $\myoptutil \le 1$.
\end{proof}

\begin{lemma}\label{lem:opt-vpb-vqa-upb-uqa}
    If $a(v,K,K,0) > R$, $a(v,I,K,0) \le R$, $a(u,K,K,0) > R$, and $a(u,I,K,0) \le R$, then $x_v + x_u \ge \myoptutil$.
\end{lemma}
\begin{proof}
    By assumption, we have
    \begin{align*}
        q_v = b(v,K,K,0) &= R, \\
        q_u = b(u,K,K,0) &= R.
    \end{align*}
    Lemma~\ref{lem:K-adj} implies $R < a(u,K,K,0) \le a(v,I,K,0) + 1 \le R + 1$; therefore $a(v,I,K,0) = R$. Then there is a $(u,K,K,0)$-walk of \Klength{} $R+1$, which implies $A(u,K) > R$. We have $b(v,I,K,0) = R$ and $B(u,K) = R$. Exchanging the roles of $v$ and $u$, also $b(u,I,K,0) = R$ and $B(v,K) = R$. Therefore
    \[
        p_v = p_u = R.
    \]
    We conclude that $x_v + x_u = 1/2 + 1/2$. Because we have $a(v,I,K,0) \le R$, there exists a $0$-vertex incident to a \Kedge{} and therefore $\myoptutil \le 1$.
\end{proof}

\begin{lemma}\label{lem:opt-vpb-vqa-upb-uqb}
    If $a(v,K,K,0) > R$, $a(v,I,K,0) \le R$, $a(u,K,K,0) > R$, and $a(u,I,K,0) > R$, then $x_v + x_u \ge \myoptutil$.
\end{lemma}
\begin{proof}
    By assumption, we have
    \[
        q_v = b(v,K,K,0) = R.
    \]
    As $b(u,K,K,0) = R$ and $b(u,I,K,0) = R$, we have $p_u = B(u,K)$ and $q_u = B(u,I)$. By the same argument as in the proof of Lemma~\ref{lem:opt-vpb-vqa-upb-uqa}, we can conclude that $b(v,I,K,0) = R$ and $B(u,K) = R$. Therefore
    \[
        p_u = R
    \]
    and $p_v = B(v,K)$. By Lemma~\ref{lem:K-adj},
    \[
        p_v = B(v,K) \ge B(u,I) = q_u.
    \]
    Therefore
    \[
        x_v + x_u
            = \frac{p_v}{p_v + R} + \frac{R}{R + q_u}
            \ge \frac{q_u}{q_u + R} + \frac{R}{R + q_u} = 1.
    \]
    Again, there exists a $0$-vertex incident to a \Kedge{} and therefore $\myoptutil \le 1$.
\end{proof}

\begin{lemma}\label{lem:opt-vpb-vqb-upb-uqb-a}
    If $a(v,K,K,0) > R$, $a(v,I,K,0) > R$, $a(u,K,K,0) > R$, $a(u,I,K,0) > R$, and $A(v,I) \ge R$, then $x_v + x_u \ge \alpha\myoptutil$.
\end{lemma}
\begin{proof}
    By assumption, $b(v,K,K,0) = R$, $b(v,I,K,0) = R$, $b(u,K,K,0) = R$, and $b(u,I,K,0) = R$. Therefore $p_v = B(v,K)$, $q_v = B(v,I)$, $p_u = B(u,K)$, and $q_u = B(u,I)$. Lemma~\ref{lem:K-adj} implies
    \begin{align*}
        {p_v = B(v,K)} &\ge {B(u,I) = q_u}, \\
        {p_u = B(u,K)} &\ge {B(v,I) = q_v}.
    \end{align*}
    Therefore
    \[
        x_v + x_u
            \ge \frac{q_u}{q_u + q_v} + \frac{q_v}{q_v + q_u} = 1.
    \]
    As $A(v,I) \ge R$, there exists a $(v,I,K,\cdot)$-walk of \Klength{} at least $R$. By Corollary~\ref{cor:long-walk}, $\myoptutil \le 1 + 1/R$. Therefore $x_v + x_u \ge \alpha \myoptutil$.
\end{proof}

\begin{lemma}\label{lem:opt-vpb-vqb-upb-uqb-b}
    If $a(v,K,K,0) > R$, $a(v,I,K,0) > R$, $a(u,K,K,0) > R$, $a(u,I,K,0) > R$, $A(v,I) < R$, and $A(u,I) < R$, then $x_v + x_u \ge \myoptutil$.
\end{lemma}
\begin{proof}
    By assumption, $b(v,K,K,0) = R$, $b(v,I,K,0) = R$, $b(u,K,K,0) = R$, and $b(u,I,K,0) = R$. Therefore $p_v = B(v,K)$, $q_v = B(v,I) = A(v,I)$, $p_u = B(u,K)$, and $q_u = B(u,I) = A(u,I)$. By Lemma~\ref{lem:K-adj}, $A(v,K) \ge A(u,I) + 1$. As $R \ge A(u,I) + 1$, we also have $B(v,K) \ge A(u,I) + 1$. Therefore
    \[
        p_v = B(v,K) \ge q_u + 1.
    \]
    By symmetry,
    \[
        p_u \ge q_v + 1.
    \]
    We conclude that
    \begin{align*}
        x_v + x_u
            &\ge \frac{q_u + 1}{(q_u + 1) + q_v} + \frac{q_v + 1}{(q_v + 1) + q_u} \\
            &= 1 + \frac{1}{q_u + q_v + 1}.
    \end{align*}
    There is a $(\cdot,K,K,\cdot)$-walk of \Klength{} $A(u,I) + A(v,I) + 1 = q_u + q_v + 1$. Because $a(v,I,K,0) > R$ and $a(u,I,K,0) > R$, both endpoints of this walk are $x$-vertices. Hence, the walk can be extended into a $(\cdot,I,K,\cdot)$-walk of the same \Klength{}. By Corollary~\ref{cor:long-walk}, ${x_v + x_u \ge \myoptutil}$.
\end{proof}

\begin{corollary}\label{cor:opt}
    If $v$ and $u$ are \Kadjacent{} $x$-vertices, then $x_v + x_u \ge \alpha\myoptutil$.
\end{corollary}
\begin{proof}
    Apply Lemmata~\ref{lem:opt-vpa-vqa}--\ref{lem:opt-vpb-vqb-upb-uqb-b} and the symmetry of $v$ and~$u$.
\end{proof}

\newpage

\section*{Acknowledgements}

We thank Valentin Polishchuk for comments and discussions. This research was supported in part by the Academy of Finland, Grants 116547 and 117499, and by Helsinki Graduate School in Computer Science and Engineering (Hecse).

\providecommand{\noopsort}[1]{}

\end{document}